\documentclass[conference,a4paper,onecolumn]{IEEEtran}

\usepackage{verbatim}
\usepackage[english]{babel}  
\usepackage{amsmath,amssymb} 
\usepackage{graphicx}        
\usepackage{color}           
\usepackage{varioref}        
\usepackage{xspace}          
\usepackage{url}             
\usepackage{cite}            
\usepackage{multirow}        
\usepackage{ifpdf}

\ifpdf
\usepackage[final,tracking=true,kerning=true,spacing=true,factor=1100,stretch=10,shrink=10]{microtype}
\fi

\newtheorem{theorem}{Theorem}
\newtheorem{lemma}[theorem]{Lemma}

\newtheorem{definition}[theorem]{Definition}

\newenvironment{proofSketch}{%
  \par\noindent\hspace{2em}{\itshape Proof sketch: }%
}{%
  \hspace*{\fill}~\IEEEQED\par%
}

\AtBeginDocument{}

\newcommand\half{\tfrac 1 2}
\newcommand\defeq{\triangleq} 

\newcommand\modop{\ {\rm mod}\ } 
\newcommand{\mo}{{-1}} 

\newcommand\F{\mathbb F\xspace} 

\newcommand\ZZ{\mathbb Z\xspace}
\newcommand\NN{\mathbb N\xspace}

\newcommand{\LT}[1]{\textnormal{\footnotesize LT}(#1)}
\newcommand{\LP}[1]{\textnormal{\footnotesize LP}(#1)}
\newcommand{\OD}[1]{{\operatorname{D}(#1)}}

\newcommand{\params}[3]{[#1,\ #2,\ #3]}

\usepackage{bm}

\usepackage[noend]{algpseudocode}
\usepackage{algorithm}
\usepackage{algorithmicx}
\algrenewcommand\alglinenumber[1]{{\scriptsize#1}}
\algrenewcommand\algorithmicrequire{\textbf{Input:}}
\algrenewcommand\algorithmicensure{\textbf{Output:}}
\newcommand{\Ifline}[2]{\State \textbf{if }#1{ \textbf{then} }#2}

\renewcommand{\vec}[1]{\bm{#1}}
\newcommand{\Mod}[1]{\mathcal{#1}}
\newcommand{\maxdeg}[1]{{\rm maxdeg}\,#1}
\renewcommand{\OD}[1]{{\Delta(#1)}}

\newcommand{\Wmap}{\Phi}
\newcommand{\matsize}{{(\ell+1)\times(\ell+1)}}
\newcommand\val{\psi} 

\newcommand{\shiftbox}[2]{\hspace{#1}#2\hspace{-#1}}

\begin{document}
\sloppy

\title{Generalised Multi-sequence Shift-Register Synthesis using Module Minimisation} 

\author{
  \IEEEauthorblockN{Johan S.~R.~Nielsen}
  \IEEEauthorblockA{Department of Applied Mathematics and Computer Science, Technical University of Denmark\\
    Email: jsrn@jsrn.dk} 
}


\maketitle
\begin{abstract}
  We show how to solve a generalised version of the Multi-sequence Linear Feedback Shift-Register (MLFSR) problem using minimisation of free modules over $\F[x]$.
  We show how two existing algorithms for minimising such modules run particularly fast on these instances.
  Furthermore, we show how one of them can be made even faster for our use.
  With our modelling of the problem, classical algebraic results tremendously simplify arguing about the algorithms.
  For the non-generalised MLFSR, these algorithms are as fast as what is currently known.
  We then use our generalised MLFSR to give a new fast decoding algorithm for Reed Solomon codes.
\end{abstract}

\section{Introduction}
\label{sec:introduction}

The Multi-sequence Linear Feedback Shift-Register (MLFSR) synthesis problem has many practical applications in fields such as coding theory, cryptography and systems theory, see e.g.~the references in \cite{schmidt06}.
The problem can be formulated as follows: over some field $\F$, given $\ell$ polynomials $S_1(x),\ldots,S_\ell(x) \in \F[x]$ and $\ell$ ``lengths'' $m_1,\ldots,m_\ell \in \ZZ_+$, find a lowest-degree polynomial $\Lambda(x)$ such that there exists polynomials $\Omega_1(x), \ldots, \Omega_\ell(x)$ satisfying 
\begin{IEEEeqnarray*}{rCl+l}
  \Lambda(x)S_i(x) &\equiv& \Omega_i(x) \mod x^{m_i} \\
  \deg \Lambda &>& \deg \Omega_i   , & i=1,\ldots,\ell
\end{IEEEeqnarray*}
Several algorithms exist for solving this problem, some using Divide \& Conquer (D\&C) techniques and some not.
Of the latter sort, the fastest have running time $O(\ell m^2)$, where $m = \max\{m_i\}$: Schmidt and Sidorenko's corrected version of Feng and Tzeng's Berlekamp--Massey generalisation \cite{schmidt06,fengTzeng}; as well as Wang et al.'s lattice minimisation approach \cite{wang08}.
The best DC algorithm is Sidorenko and Bossert's variant of the corrected Feng--Tzeng BMA \cite{sidorenko11skew} and has running time $O(\ell^3 m\log^2 m\log\log m)$.
Obviously, whichever is fastest depends on the relative size of $\ell$ and $m$.

In this paper, we give algorithms that solve the following natural generalisation (MgLFSR): given $S_1(x),\ldots,S_\ell(x) \in \F[x]$, moduli $G_1(x),\ldots,G_\ell(x) \in \F[x]$ as well as weights $\nu \in \ZZ_+$ and $w_0,\ldots,w_\ell \in \NN_0$, find a lowest-degree polynomial $\Lambda(x)$ such that there exist polynomials $\Omega_1(x),\ldots, \Omega_\ell(x)$ satisfying
\begin{IEEEeqnarray}{rCl+l}
  \Lambda(x)S_i(x) &\equiv& \Omega_i(x) \mod G_i(x)  \notag \\
  \nu\deg \Lambda + w_0 &>& \nu\deg \Omega_i + w_i , & i=1,\ldots,\ell \label{eqn:mglfsr}
\end{IEEEeqnarray}
We model the above problem as that of finding a ``minimal'' vector in a certain free $\F[x]$ module.
Such a vector can be found as an element of any basis of the module which satisfies certain minimality properties, and standard algorithms in the literature can compute such a basis.
We describe the Mulders--Storjohann algorithm \cite{mulders03} and give an improved complexity analysis for our case, arriving at the running time $O(\ell^2 m^2)$, where $m = \max_i\{\deg G_i + w_i\nu^\mo\}$.
We then demonstrate how this algorithm is amenable to two distinct speed-ups:
\begin{itemize}
  \item A D\&C variant achieving $O(\ell^3 m\log^2 m\log\log m)$; for general module reduction, this algorithm is known as Alekhnovich's \cite{alekhnovich05}, but we point out it is a variant of Mulders--Storjohann.
  \item A new demand-driven variant utilising the special form of the module of the MgLFSR to achieve complexity $O(\ell m \tilde P(m))$, where $\tilde P(m) = m$ if all $G_i$ are sparse (in particular, if they are powers of $x$) and $\tilde P(m) = m\log m\log\log m$ otherwise.
\end{itemize}
These complexities match the best known ones for the MLFSR case.
Our approach draws much inspiration from Fitzpatrick's module view on the classic Key Equation \cite{fitzpatrick95}, and can be seen as a natural extension to this.
Though our initial aim was a solution to the MLFSR, the MgLFSR emerged as generalisations also easy handled; in Section \ref{sec:decoding} we give an application of this generality with a new algorithm for decoding Reed Solomon codes beyond half the minimum distance.

\section{Preliminaries}
\label{sec:preliminiaries}

\subsection{Notation}

In the sequel, we will refer to the $S_i$, $G_i$ as well as the weights $\nu,w_0,\ldots, w_\ell$ as being from a particular instance of the MgLFSR.
We will use the term ``\emph{solution}'' of this MgLFSR for any vector $(\lambda, \omega_1,\ldots,\omega_\ell) \in \F[x]^{\ell+1}$ which satisfies the equations of \eqref{eqn:mglfsr} for $i=1,\ldots,\ell$; a solution where $\deg \lambda$ is minimal is called a \emph{minimal solution}, and we seek one such.

We will assume that $\deg S_i < \deg G_i$ for each $i$; for otherwise replacing $S_i$ with $(S_i \modop G_i)$ admits exactly the same solutions to the MgLFSR. 
We also assume that $w_0 < \max_i\{\deg S_i + w_i\}$ since otherwise $(1,S_1,\ldots,S_\ell)$ is the minimal solution.

We extensively deal with vectors and matrices over $\F[x]$.
We use the following notational conventions:
\begin{itemize}
  \item A matrix is named uppercase: $V$. Rows use the same letter lowercase and indexed: $\vec v_i$.
  If $\vec v$ is a vector, then $v_j$ are its elements; the cells of matrices have double subscripts: $v_{i,j}$.
  We'll use zero-indexing, so if $\vec v$ has length $\ell+1$ its elements are $v_0,\ldots,v_\ell$.
  \item The degree of a non-zero vector $\vec v$ is $\deg \vec v = \max_i\{\deg v_i\}$.
  The degree of a matrix $V$ is $\deg V = \sum_{i} \deg \vec v_i$.
  The \emph{max-degree} of $V$ is $\maxdeg V = \max_{i,j}\{ \deg v_{i,j} \}$.
  \item Let the \emph{leading position} of a non-zero vector $\vec v$ be $\LP{\vec v} = \max\{ j \mid \deg v_j = \deg \vec v \}$. The \emph{leading term} is the polynomial $\LT{\vec v} = v_{\LP{\vec v}}$.
\end{itemize}

In complexity estimates, we will let $m = \max_i\{ \deg G_i + \frac {w_i}{\nu} \}$.
$P(m)$ will be the cost of multiplying two polynomials of degree at most $m$; we can set $P(m) = m\log m\log\log m$ using Sch\"onhage-Strassen, see e.g.~\cite[Theorem 8.23]{gathen}.

\subsection{Satisfying the congruence equations}

We can consider the space $\Mod M$ of all vectors $(\lambda, \omega_1, \ldots,\omega_\ell) \in \F[x]^{\ell+1}$ such that $\lambda S_i \equiv \omega_i \mod G_i$ for $i = 1,\ldots,\ell$.
Solutions are thus those vectors such that $\nu\deg \lambda + w_0 > \nu\deg \omega_i + w_i$.

All restrictions defining $\Mod M$ are $\F[x]$-linear so $\Mod M$ is a module over $\F[x]$.
Shortly, we'll see that $\Mod M$ is free, so any finite basis can be represented as a matrix where each row corresponds to a basis element.
We will in the sequel often simply say ``basis'' for such a matrix representation.
\begin{lemma}
The following is a basis for $\Mod M$:
\[
  M = 
  \left[
    \begin{matrix}
       1 & S_1 & S_2 & \cdots & S_\ell \\
         & G_1 &     & \raisebox{-1ex}{\shiftbox{1em}{\multirow{2}{*}{\large 0}}}   &       \\
         &     & G_2 &        &       \\
         & \raisebox{1ex}{\multirow{2}{*}{\large 0}}   &     & \ddots &       \\
         &     &     &        & G_\ell
    \end{matrix}
  \right]
\]
\end{lemma}
\begin{proof}
  Clearly, each row of $M$ is in $\Mod M$.
  Since any vector $\vec v \in \Mod M$ satisfies $v_0S_i \equiv v_i \mod G_i$ for $i=1,\ldots,\ell$, it means there exists some $p_1,\ldots,p_\ell \in \F[x]$ such that $v_i = v_0S_i + p_iG_i$.
  Therefore $\vec v = v_0\vec m_0 + p_1\vec m_1 + \ldots + p_\ell \vec m_\ell$.
\end{proof}

To deal with the weights of the MgLFSR in an easy manner, we will introduce a mapping which will ``embed'' the weights into the basis.
Define $\Wmap: \F[x]^{\ell+1} \rightarrow \F[x]^{\ell+1}$ by
\[
  \big(a_0(x),\ldots,a_\ell(x) \big) \mapsto
    \big(x^{w_0}a_0(x^{\nu}), \ldots, x^{w_\ell}a_\ell(x^{\nu}) \big)
\]
In the case of MLFSR, $\Wmap$ is simply the identity function.
Extend $\Wmap$ element-wise to sets of vectors, and extend $\Wmap$ row-wise to $\matsize$ matrices such that the $i$th row of $\Wmap(V)$ is $\Wmap(\vec v_i)$.
Note that $\Wmap(\Mod M)$ is a free $\F[x^{\nu}]$-module of dimension $\ell+1$, and that any basis of it is by $\Wmap^\mo$ sent back to a basis of $\Mod M$.
\begin{lemma}
  \label{lem:minSolWeighted}
  A non-zero $\vec s \in \Mod M$ is a minimal solution to the MgLFSR if and only if $\LP{\Wmap(\vec s)} = 0$ and for all non-zero $\Wmap(\vec b) \in \Wmap(\Mod M)$ with $\LP{\Wmap(\vec b)}=0$ it holds that $\deg \Wmap(\vec s) \leq \deg \Wmap(\vec b)$.
\end{lemma}
\begin{proof}
  $\vec s = (\Lambda,\Omega_1,\ldots,\Omega_\ell) \in \Mod M$ is a solution to the MgLFSR if and only if $\LP{\Wmap(\vec s)} = 0$, since $\nu\deg \Lambda + w_0 > \nu\deg \Omega_i + w_i \iff \deg (x^{w_0}\Lambda(x^{\nu})) > \deg(x^{w_i}\Omega_i(x^{\nu}))$.
  Since $\deg \Wmap(\vec b) = \nu\deg b_0$ whenever $\LP{\Wmap(\vec b)}=0$, $\vec s$ is then a minimal solution if and only if $\deg(\Wmap(\vec s))$ is minimal for vectors in $\Wmap(\Mod M)$ with leading position 0.
\end{proof}

\subsection{Module minimisation}

\begin{definition}
  \label{def:weakpopov}
  A full-rank matrix $V \in \F[x]^\matsize$ is in \emph{weak Popov} form if an only if the leading position of all rows are different.
  The \emph{orthogonality defect} of $V$ is $\OD V \defeq \deg V - \deg \det V$.
\end{definition}
The concept of orthogonality defect was introduced by Lenstra \cite{lenstra85} for estimating the running time of his algorithm on module minimisation; we will use it to a similar effect.
The following lemma gives the foundations for such a use; we omit the proof which can be found in \cite{nielsen13}:
\begin{lemma}[\protect{\cite[Lemma 11]{nielsen13}}]
  \label{lem_popovreduces}
  If a matrix $V$ over $\F[x]$ is in weak Popov form then $\OD V = 0$.
\end{lemma}
Note that for any square matrix $V$, $\OD V \geq 0$; thus since the determinant is the same for \emph{any} basis of the module for which $V$ is a basis, $\OD V$ measures how much $\deg V$ is greater than the minimal degree possible.
Due to its special form, $M$ has particularly low orthogonality defect:
\begin{lemma}
  \label{lem:Mortho}
  $\OD{\Wmap(M)} = \max\{ w_i + \nu\deg S_i(x) \} - w_0 \leq \nu m - w_0$.
\end{lemma}
\begin{proof}
  Since $M$ is upper triangular $\det (\Wmap(M)) = x^{w_0}\prod_{i=1}^\ell x^{w_i}G_i(x^{\nu})$ and the lemma follows.
\end{proof}
Lastly, a crucial property of matrices in weak Popov form:
\begin{lemma}
  \label{lem:weakGivesMin}
  Let $V \in \F[x]^\matsize$ be a basis in weak Popov form of a module $\Mod V$.
  Any non-zero $\vec b \in \Mod V$ satisfies $\deg \vec v \leq \deg \vec b$ where $\vec v$ is the row of $V$ with $\LP{\vec v} = \LP{\vec b}$.
\end{lemma}
\begin{proof}
  Since $V$ is a basis of $\Mod V$, there exists $p_0, \ldots, p_\ell \in \F[x]$ such that $\vec b = p_0\vec v_1 + \ldots + p_\ell \vec v_\ell$ where the $\vec v_i$ are rows of the $V$.
  But the $\vec v_i$ all have different leading position, so the $p_i\vec v_i$ must as well for those $p_i \neq 0$.
  Therefore, there is exactly one $j$ such that $\LP{\vec b} = \LP{v_j}= \LP {p_j\vec v_j}$ and $\deg \vec b = \deg (p_j\vec v_j) = \deg p_j + \deg \vec v_j$.
\end{proof}

Combining Lemma \ref{lem:minSolWeighted} and Lemma \ref{lem:weakGivesMin} we see that a basis for $\Wmap(\Mod M)$ in weak Popov form must contain a row $\Wmap(\vec s)$ such that $\vec s$ is a minimal solution to the MgLFSR.
Any algorithm which brings $\F[x]$-matrices to weak Popov form can thus be used to solve the MgLFSR.

\section{Simple minimisation}
\label{sec:simple_minimisation}

\begin{definition}
  \label{def:rowred} 
  Applying a \emph{row reduction} on a full-rank matrix over $\F[x]$ means to find two different rows $\vec v_i, \vec v_j$, $\deg \vec v_i \leq \deg \vec v_j$ such that $\LP{\vec v_i} = \LP{\vec v_j}$, and then replacing $\vec v_j$ with $\vec v_j - \alpha x^\delta \vec v_i$ where $\alpha \in \F$ and $\delta \in \NN_0$ are chosen such that the leading term of the polynomial $\LT{\vec v_j}$ is cancelled.
\end{definition}

Define a \emph{value} function for vectors $\val: \F[x]^{\ell+1} \rightarrow \NN_0$:
\begin{equation}
  \label{eqn:value}
  \val(\vec v) = (\ell+1)\deg \vec v + \LP{\vec v}
\end{equation}
\begin{lemma}
  \label{lem:rowreddec}
  If $\vec v_j'$ is the vector replacing $\vec v_j$ in a row reduction, then $\val(\vec v_j') < \val(\vec v_j)$.
\end{lemma}
\begin{proof}
  We can't have $\deg \vec v_j' > \deg \vec v_j$ since all terms of both $\vec v_j$ and $\alpha x^\delta \vec v_i$ have degree at most $\deg \vec v_j$.
  If $\deg \vec v_j' < \deg \vec v_j$ we are done since $\LP{\vec v_j'} < \ell+1$, so assume $\deg \vec v_j' = \deg \vec v_j$.
  Let $h = \LP{\vec v_j} = \LP{\vec v_i}$.
  By the definition of $\LP{\cdot}$, all terms in both $\vec v_j$ and $\alpha x^\delta \vec v_i$ to the right of $h$ must have degree less than $\deg \vec v_j$, and so also all terms in $\vec v_j'$ to the right of $h$ satisfies this.
  The row reduction ensures that $\deg v'_{j,h} < \deg v_{j,h}$, so it must then be the case that $\LP{\vec v_j'} < h$.
\end{proof}

The following elegant algorithm for general $\F[x]$-module minimisation is due to Mulders and Storjohann \cite{mulders03}.
Correctness and complexity is established in Lemma \ref{lem:storjohann}, whose proof is modelled over the proof in \cite{mulders03} but specialised for input of the form of $\Wmap(M)$.
\begin{algorithm}
  \caption{Mulders--Storjohann}
  \label{alg:storjohann}
  \begin{algorithmic}[1]
    \Require{$V = \Wmap(M)$.}
    \Ensure{A basis of $\Wmap(\Mod M)$ in weak Popov form.}
    \State Apply row reductions on the rows of $V$ until no longer possible.
    \State \Return $V$.
  \end{algorithmic}
\end{algorithm}

\begin{lemma}
  \label{lem:storjohann}
  Algorithm~\ref{alg:storjohann} is correct.
  It performs less than $(\ell+1)(m-w_0\nu^\mo+2)$ row reductions and has asymptotic complexity $O(\ell^2m^2)$.
\end{lemma}
\begin{proof}
  Since the row reductions are performed over $\F[x]$, we first need to argue that we do not leave the $\F[x^{\nu}]$ module for $V$ to continue to be a basis of $\Wmap(\Mod M)$ after each row reduction: however, since any $\vec u, \vec v \in \Wmap(\Mod M)$ have $\deg u_i \equiv \deg v_i \mod \nu$ for all $i$, the $x^\delta$ scalar in each row reduction is a power of $x^{\nu}$; thus, they are indeed $\F[x^{\nu}]$ row reductions. 
  Since we can apply a row reduction on a matrix if and only if it is not in weak Popov form, the algorithm must bring $V$ to weak Popov form in case it terminates.

  Termination follows directly from Lemma \ref{lem:rowreddec} since the value of a row decreases each time a row reduction is performed.
  We can be more precise, though.
  For any non-zero $\vec v \in \Mod M$:
  \begin{IEEEeqnarray*}{rCl}
    \val(\Wmap(\vec v)) &=& (\ell+1)(\nu\deg v_{\LP{\Wmap(\vec v)}} + w_{\LP{\Wmap(\vec v)}}) + \LP{\Wmap(\vec v)}
    \\   &\equiv& (\ell+1)w_{\LP{\Wmap(\vec v)}} + \LP{\Wmap(\vec v)} \mod (\ell+1)\nu
  \end{IEEEeqnarray*}
  So on any given interval of size $(\ell+1)\nu$, $\val(\Wmap(\vec v))$ can attain at most $\ell+1$ of the values, depending on its leading position.
  Denote now by $\Wmap(U)$ the matrix in weak Popov form returned by the algorithm.
  Due to the above, the algorithm will perform a row reduction on the $i$th row at most $\big\lceil \frac {\ell+1}{(\ell+1)\nu}\big(\val(\Wmap(\vec m_i))-\val(\Wmap(\vec u_i))\big) \big\rceil$ times.
  Since $\deg(\Wmap(U)) = \deg \det(\Wmap(U)) = \deg \det(\Wmap(M))$ and the $\LP{\Wmap(\vec u_i)}$ are all different, 
  the total number of row reductions is then upper bounded by:
  \begin{IEEEeqnarray*}{l}
    {\textstyle \sum_{i=0}^\ell\big\lceil \nu^\mo\big(\val(\Wmap(\vec m_i))-\val(\Wmap(\vec u_i))\big) \big\rceil }
    \\ \leq {\textstyle  \ell+1 + \frac{\ell+1}{\nu}(\deg(\Wmap(M)) - \deg(\Wmap(U))) + \LP{\Wmap(\vec m_0)}}
    \\ \leq \tfrac{\ell+1}{\nu}\OD{\Wmap(M)} + 2\ell + 1
        \IEEEyesnumber \label{eqn:storjohann:rowreds}
  \end{IEEEeqnarray*}

  For the asymptotic complexity, note that during the algorithm, no polynomial in $V$ will have larger degree than $\maxdeg(\Wmap(M)) = \nu m$.
  Since the polynomials in $\Wmap(\Mod M)$ are sparse with only every $\nu$th coefficient non-zero, they can be represented and manipulated as fast as usual polynomials of degree $m$.
  One row reduction consists of $\ell+1$ times scaling and adding two such polynomials.
\end{proof}

\section{The Divide \& Conquer speed-up}
\label{sec:dc}

Algorithm \ref{alg:storjohann} admits a D\&C version which is due to Alekhnovich \cite{alekhnovich05}.
However, he seemed not to be aware of the work of Mulders and Storjohann, and that his algorithm is indeed a variant of theirs.
Since all the formal results we need are in \cite{alekhnovich05}---as well as the more general analysis in \cite{brander10}---we will here only give an overview of the algorithm and its connection to Algorithm \ref{alg:storjohann}, as well as the complexity result.

The algorithm works by structuring its row reductions in a tree-like fashion; more precisely it hinges on the following series of observations, all of which are proved in \cite{alekhnovich05}:
\begin{enumerate}
  \item Imagine the row reductions bundled such that each bundle reduces $\maxdeg V$ by 1, where $V$ is the result of applying all earlier row reductions to the input.
  \item To calculate the row reductions in one such bundle on $V$, one needs for each row $\vec v_i$ of $V$ to know only the monomials in $\vec v_i$ having degree $\deg \vec v_i$.
  \item Therefore, to calculate a series of $t$ such bundles, one needs to know only monomials of degree greater than $\deg \vec v_i-t$. Call the matrix containing only these a $t$-projection of $V$.
  \item Any series of row reductions can be represented as a matrix $U \in \F[x]^\matsize$ where the product $UV$ is then the result of applying those row reductions to $V$.
  \item Thus, we can structure the bundles in a binary tree, where to calculate the row reduction matrix for some node, representing say $t$ bundles, given the matrix $V$, one first recursively calculates the left half of the bundles on a $t/2$-projection of $V$ to get a row reduction matrix $U_1$.
  Then recursively calculate the right half of the bundles on a $t/2$-projection of $U_1V$ to get $U_2$, and the total row reduction matrix becomes $U_2U_1$.
\end{enumerate}
We have exactly the same choice of row reductions as in Algorithm \ref{alg:storjohann}, but the computations are now done on matrices where each cell contains only one monomial (since, in the leaves of the tree, we work on 1-projections), speeding up those calculations by a factor $m$.
Collecting the row reductions is then done using matrix multiplications.

That Alekhnovich's algorithm can bring $\Wmap(M)$ to weak Popov form follows immediately from its general correctness; however, for a better estimate on its running time, we need to correctly consider the effects of weights.
This is not done in \cite{alekhnovich05}, but it was done by Brander in \cite{brander10}.
With observations similar to those in Section \ref{sec:simple_minimisation}, for our case we get:
\begin{lemma}
  Alekhnovich's algorithm on $\Wmap(M)$ has asymptotic complexity $O(\ell^3P(m)\log m)$
\end{lemma}
\begin{proof}
  Inserting into \cite[Theorem 3.14]{brander10}, we get complexity $O(\ell^2t)$ for computing the row reductions, added with $O(\ell^3P(\nu^\mo t) \log(t))$ for all matrix multiplications, where $t = \deg(\Wmap(M)) - \deg(\Wmap(U))$ and $\Wmap(U)$ is the output.
  In our case, we set $t = \OD{\Wmap(M)} \leq \nu m$ to ensure $\Wmap(U)$ is in weak Popov form.
  However, Brander used in both estimates that the number of row reductions was bounded by $O(\ell t)$; since we showed in Lemma \ref{lem:storjohann} that it was indeed only $O(\ell m)$, we can compute the row reductions in only $O(\ell^2 m)$ and the matrix multiplications in $O(\ell^3P(m)\log m)$.
\end{proof}

\section{The Demand-Driven speed-up}
\label{sec:bma}

\def\curdeg{\theta}
\def\p#1{\tilde #1}
\def\prev{\mathsf{previous}}

We will show how to obtain a faster variant of \mbox{Algorithm \ref{alg:storjohann}} using the following observation: it is essentially sufficient to keep track of only the first column of $V$ during the algorithm, and then calculate the other entries when the need arise.
The resulting algorithm bears a striking resemblance to the Berlekamp-Massey for MLFSR \cite{schmidt06}, though of course the manner in which these algorithms are obtained differs.

Overload $\val$ to $\NN_0 \times \{0,\ldots,\ell\} \rightarrow \NN_0$ by $\val(\curdeg,i) = (\ell+1)\curdeg + i$, i.e. for any non-zero $\vec v \in \F[x]^{\ell+1}$, $\val(\vec v) = \val(\deg \vec v, \LP{\vec v})$. 
Define the helper function
\begin{IEEEeqnarray*}{rCl.l}
\textstyle
\prev(\curdeg,i) &=& \arg\max_{\curdeg', i'}\{ \val(\curdeg', i') &\mid \val(\curdeg',i') < \val(\curdeg,i) 
                 \  \land\  \curdeg' \equiv w_{i'} \mod \nu \}
\end{IEEEeqnarray*}
$\prev$ gives the degree and leading position a vector in $\Wmap(\Mod M)$ should have for attaining the greatest possible $\val$-value less than $\val(\curdeg,i)$.

\begin{algorithm}
  \caption{Demand--Driven MgLFSR Minimisation}
  \label{alg:bma}
  \begin{algorithmic}[1]
    \def\continue{\mathsf{continue}}
    \Require{$\p S_i =x^{w_i}S_i(x^\nu),\ \p G_i = x^{w_i}G_i(x^\nu)$ for $i=1,\ldots,\ell$\;}
    \Ensure{$\Lambda(x)$, a minimal solution to the MgLFSR\;}
    \State $(\curdeg, i) = (\deg, \textnormal{\footnotesize LP})$ of $(x^{w_0}, \p S_1, \ldots, \p S_\ell)$
    \Ifline{$i = 0$}{\Return $1$}
    \State $(\lambda_0,\ldots,\lambda_\ell) = (x^{w_0},0,\ldots,0)$
    \State $\alpha_jx^{\curdeg_j} = $ the leading monomial of $\p G_j$ for $j=1,\ldots,\ell$
    \While{$\deg \lambda_0 \leq \curdeg$}
      \State $\alpha = $ coefficient to $x^\curdeg$ in $(x^{-w_0}\lambda_0\p S_i \mod \p G_i)$
         \label{alg:bma:mono_ass}
      \If{$\alpha \neq 0$}
        \Ifline{$\curdeg < \curdeg_i$}{swap $(\lambda_0,\alpha,\curdeg)$ and $(\lambda_i, \alpha_i, \curdeg_i)$}
        \State $\lambda_0 = \lambda_0 - \frac \alpha {\alpha_i} x^{\curdeg-\curdeg_i} \lambda_i$
          \label{alg:bma:upd_lambda}
      \EndIf
      \State $(\curdeg, i) = \prev(\curdeg, i)$ \label{alg:bma:prev}
      \Ifline{$i = 0$}{$(\curdeg, i) = \prev(\curdeg, i)$}
    \EndWhile
    \State \Return $x^{-w_0}\lambda_0(x^{1/\nu})$
  \end{algorithmic}
\end{algorithm}

We will prove the correctness of the algorithm by showing that the computations correspond to a possible run of a slight variant of Algorithm \ref{alg:storjohann}; first we need a technical lemma:
\begin{lemma}
  \label{lem:storjohannVar}
  Consider a variant of Algorithm \ref{alg:storjohann} where we, when replacing some $\vec v_j$ with $\vec v'_j$ in a row reduction, instead replace it with $\vec v''_j = (v'_{j,0}, v'_{j,1} \modop \p G_1, \ldots, v'_{j,\ell} \modop \p G_\ell)$.
  This does not change correctness of the algorithm or the upper bound on the number of row reductions performed.
\end{lemma}
  \begin{proof}
    Correctness follows if we can show that each of the $\ell$ modulo reductions could have been achieved by a series of $\F[x^{\nu}]$ row operations on the current matrix $V$ after the row reduction producing $\vec v'$, since then $V$ would remain a basis of $\Wmap(\Mod M)$.

    Consider the modulo reduction on the $h$th position.
    This could be achieved by adding a multiple of $\vec g_h = (0,\ldots,0,\p G_h,\ldots,0)$, with position $h$ non-zero, to $\vec v'$.
    That this multiple is in $\F[x^\nu]$ follows from the fact that any $\vec u \in \Wmap(\Mod M)$ has $\deg u_h \equiv w_h \mod \nu$.
    Since $\vec g_h$  is a row in $\Wmap(M)$, then as long as this has not yet been row reduced, the $h$th position reduction is allowed.
    Notice that if $\val(\vec v') < \val(\vec g_h)$ then the reduction using $\vec g_h$ is void.

    Introduce now a loop invariant involving $J_h = \{ \vec g_h \}$, a subset of the current rows in $V$ having two properties: that $\vec g_h$ can be constructed as an $\F[x^{\nu}]$-linear combination of the rows in $J_h$; and that each $\vec v \in J_h$ has $ \val(\vec v) \leq \val(\vec g_h)$.
    After row reductions on rows not in $J_h$, the $h$th modulo reduction is therefore allowed, since $\vec g_h$ can be constructed by the rows in $J_h$.
    On the other hand, after a row reduction on a row $\vec v \in J_h$ by some $\vec v_k$ resulting in $\vec v'$, the $h$th modulo reduction has no effect since $\val(\vec v') < \val(\vec v) \leq \val(\vec g_h)$.
    Afterwards, $J_h$ is updated as $J_h = J_h \setminus \{ \vec v \} \cup \{ \vec v', \vec v_k \}$ and the loop invariant is kept since $\val(\vec v_k) \leq \val(\vec v)$.

    Since $\val(\vec v''_j) \leq \val(\vec v'_j)$ the proof of Lemma \ref{lem:storjohann} shows that the number of row reductions performed is not worse than in Algorithm \ref{alg:storjohann}.
  \end{proof}
\begin{lemma}
  Algorithm \ref{alg:bma} is correct.
\end{lemma}
  \begin{proof}
    Let $V$ be the matrix continually changing in Lemma \ref{lem:storjohannVar}'s variant of Algorithm \ref{alg:storjohann}.  Let us say for a matrix $U$ that there is a ``conflict on $(i,j)$'' if $\LP{\vec u_i} = \LP{\vec u_j}$ and $\deg \vec u_i \leq \deg \vec u_j$, i.e. one could perform a row reduction on $\vec u_i, \vec u_j$.  Observe that initially $V = \Wmap(M)$ has exactly one conflict, and that after every row reduction, either there is only one conflict and it involves the replaced row, or there are zero conflicts and the algorithm is finished.  Thus for notational convenience, we consider a further variant of Algorithm \ref{alg:storjohann} where we possibly swap the two rows after a row reduction such that the reduced is the zeroth row afterwards.  Note furthermore that initially $\LP{\vec v_i} = i$ for $i\geq 1$, and that the above swapping strategy would keep also this invariant during Algorithm \ref{alg:storjohann}.

    We will demonstrate the following additional invariants:
    \begin{enumerate}
      \item Each iteration of Algorithm \ref{alg:bma} where lines 7--8 are run correspond to one row reduction on $V$;
      \item $(\lambda_0,\ldots,\lambda_\ell)$ will correspond to the first column of $V$; \label{lem:inv:lambdas}
      \item $\alpha_j x^{\curdeg_j}$ will be the leading monomial of $\LT{\vec v_j}$ for $j\geq 1$;\label{lem:inv:leads}
      \item $\val(\vec v_0) \leq \val(\curdeg, i)$ \label{lem:inv:deg}
    \end{enumerate}
    These invariants are clearly true after initialisation; assume now they are true on entry to the loop body, and we will show they are true on exit.  Once Algorithm \ref{alg:bma} terminates, $\vec v_0$ will have $\LP{\vec v_0} = 0$, so by Lemma \ref{lem:minSolWeighted}, $\lambda_0$ will be a minimal solution to the MgLFSR.

    In Lemma \ref{lem:storjohannVar}'s variant of Algorithm \ref{alg:storjohann}, note that for any row $\vec v$ of $V$, $v_j = (x^{-w_0}v_0\p S_j \modop \p G_j)$ for $j \geq 1$.  Thus, in Line \ref{alg:bma:mono_ass}, $\alpha x^\curdeg$ will be the leading monomial of $v_{0,i}$ if and only if $\val(\vec v_0) = \val(\curdeg, i)$, and otherwise, due to Invariant \ref{lem:inv:deg}, $\alpha = 0$.  In the latter case, we simply update $(\curdeg, i)$ to reflect our improved knowledge, keeping the invariants.

    However, if $\alpha \neq 0$, there is a conflict on $(0,i)$.
    Since there are no other conflicts, Algorithm \ref{alg:storjohann} will perform the next row reduction on this.
    We perform a swap such that the row to be updated, i.e.~the one with greatest degree, is the zeroth; this corresponds to line 7 (if their degrees are equal, there is a choice on which of rows $0$ or $i$ to reduce, and we choose $0$).  Note that Invariants \ref{lem:inv:leads}--\ref{lem:inv:deg} are still true after the possible swap.

    Algorithm \ref{alg:storjohann} will update $\vec v_0$ exactly as $\vec v'_0 = \vec v_0 - \frac \alpha {\alpha_i} \vec v_i$, due to the above and Invariant \ref{lem:inv:leads}.  Modulo reductions are possibly applied afterwards.  Thus the update in Line \ref{alg:bma:upd_lambda} ensures that $(\lambda_0,\ldots,\lambda_\ell)$ is the first column of $V$ after the row reduction.
  \end{proof}
For complexity estimates, define $\tilde P(t)$ as the complexity of calculating Line \ref{alg:bma:mono_ass} and Line \ref{alg:bma:upd_lambda}, with $t$ being the maximal degree of the in-going polynomials.
We could calculate Line \ref{alg:bma:mono_ass} as a polynomial multiplication followed by a division, so at least $\tilde P(t) \subset O(P(t))$.
However, sometimes we can do better: if all $G_i(x)$ are powers of $x$, the modulo reduction in Line \ref{alg:bma:mono_ass} is free, and we can perform the remaining computation in only $O(t)$.
In general, if the number of non-zero monomials of each $G_i(x)$ is upper-bounded by some constant, the computation can be done in $O(t)$.

\begin{lemma}
  Algorithm \ref{alg:bma} has complexity $O(\ell m \tilde P(m))$.
\end{lemma}
\begin{proof}
  Each iteration through the loop costs at most $\tilde P(m)$ since all polynomials are sparse of degree at most $\nu m$.
  Due to Invariant \ref{lem:inv:deg} and Line \ref{alg:bma:prev}, each iteration decreases the upper bound on one of $V$'s row's value.
  We counted in \eqref{eqn:storjohann:rowreds} that this can done at most $O(\ell m)$ times.
\end{proof}


\section{Power-Gao Decoding GRS codes}
\label{sec:decoding}
\def\C{\mathcal C}

Schmidt et al. demonstrated how one can decode low-rate Generalised Reed-Solomon (GRS) codes beyond half the minimum distance by solving an MLFSR \cite{schmidt06rs}.
This ``Power decoding'' works by noting that the classical Key Equation can be extended to several ones.
The resulting MLFSR problem could of course be solved using the algorithms of this paper.

We will briefly present a similar decoding strategy which instead extends what one could call the Key Equation of Gao's decoding algorithm \cite{gao02}.
It should be noted that this algorithm could also be used for decoding Interleaved GRS codes, just as the one by Schmidt et al. \cite{schmidt09,wachter12irs}.

Let $\C = \{ \big(f(\alpha_0),\ldots, f(\alpha_{n-1})\big) \mid f \in \F[x], \deg f < k \}$ be a (simple) $\params n k {d=n-k+1}$ GRS code with evaluation points $\alpha_0,\ldots,\alpha_{n-1} \in \F$.
Consider a sent codeword $\vec c \in \C$ which comes from evaluating some $f(x)$.
Let $\vec c$ be subjected to an unknown error pattern $\vec e \in \F^n$ such that $\vec r = \vec c + \vec e$ is received.
Define the (unknown) \emph{error locator} as $\Lambda(x) = \prod_{e_j \neq 0}(x-\alpha_j)$.
Define now also the known $G(x) = \prod_{j=0}^{n-1}(x-\alpha_j)$ as well as $R(x)$ by $R(\alpha_j) = r_j$ for $j=0,\ldots,n-1$.
\begin{lemma}
  \label{lem:gao}
  $\Lambda(x)R^i(x) \equiv \Lambda(x)f^i(x) \mod G(x),  i \in \ZZ_+$
\end{lemma}
\begin{proof}
  Polynomials are equivalent modulo $G(x)$ if and only if they have the same evaluation at $\alpha_0,\ldots,\alpha_{n-1}$.
  For $\alpha_i$ where $e_i \neq 0$, both sides of the above evaluate to zero, while for the remaining $\alpha_i$ they both give $\Lambda(\alpha_j)r_j^i = \Lambda(\alpha_j)c_j^i$.
\end{proof}
This leads us to consider the following MgLFSR: choose some $\ell \in \ZZ_+$.
Let $G_i = G$ and $S_i = (R^i \modop G)$, as well as $\nu = 1, w_0 = \ell(k-1)+1$ and $w_i = (\ell-i)(k-1)$ for $i=1,\ldots,\ell$.
The vector $\vec s = (\Lambda, \Lambda f, \ldots, \Lambda f^\ell)$ will then be a solution to the MgLFSR.

To find $\vec s$ with the algorithms of this paper, we need it also to be minimal, and that any other minimal solution is a constant multiple of $\vec s$.
We can estimate an upper bound on the degree of a minimal solution $(\lambda,\omega_1,\ldots,\omega_\ell)$ as follows:
since $\lambda S_i \equiv \omega_i \mod G_i$ implies that there exists $p_i \in \F[x]$ with $\deg p_i = \deg(\lambda S_i)$ such that $\lambda S_i - p_iG_i = \omega_i$, we can consider the MgLFSR as a homogeneous linear system of equations in the coefficients of $\lambda, p_1,\ldots,p_\ell$, such that $\lambda S_i - p_iG_i$ should have coefficient 0 for $x^{\lfloor \deg \Lambda - \nu^\mo(w_i-w_0) \rfloor}, \ldots, x^{\deg(\Lambda S_i)}$.
This linear system has non-zero solutions whenever $\deg \lambda \geq \frac \ell {\ell+1} n - \frac \ell {\ell+1} - \half \ell(k-1)$.
Thus, whenever the error locator $\Lambda$ has degree at least this, we cannot hope that $\vec s$ is a minimal solution.
For fewer errors than the above, we need a deeper analysis to estimate the probability that $\vec s$ is the minimal solution; such an analysis is done for the original Power decoding \cite{schmidt06rs}, where they find the same upper bound for error correction.

Using Algorithm \ref{alg:storjohann}, we could solve this MgLFSR in $O(\ell^2n^2)$, while Alekhnovich's algorithm could do it in $O(\ell^3n\log^2 n\log\log n)$.
Algorithm \ref{alg:bma} would be $O(\ell n^2\log n\log\log n)$.
One of the two latter will be fastest, but it will depend on the relation between $n$ and $\ell$.
Note that the pre-processing of calculating $R$ and $G$ can be done in $O(n\log^2 n \log\log n)$, see e.g.~ \cite[p. 235]{gathen}.

\section{Conclusion}

We have introduced the generalisation MgLFSR of the well-studied problem of synthesising shift-registers with multiple sequences, and shown how this can be modeled as that of finding ``minimal'' vectors in certain $\F[x]$ modules. 
There are off-the-shelf algorithms in the literature for solving this, and we demonstrated how a particularly simple of those---the Mulders--Storjohann algorithm \cite{mulders03}---runs faster on MgLFSR instances than on general $\F[x]$-matrices.

We then described how this algorithm is amenable to two speed-ups: firstly, a D\&C-approach leads to the known Alekhnovich's algorithm \cite{alekhnovich05} which we also showed has better than generic running time for MgLFSRs.
Secondly, by observing that for MgLFSRs we can postpone calculations in a demand-driven manner, we reach an algorithm resembling the Berlekamp-Massey for MLFSRs \cite{fengTzengEA,schmidt06}.

The two presented variants are as fast as the best existing algorithms for the usual MLFSR, but they are more flexible and have easy proofs of correctness due to the algebraic foundations from module minimisation.
The two speed-ups, unfortunately, seem incompatible.

The utility of the MgLFSR generalisation was demonstrated by a new decoding algorithm for GRS codes: a variant of the ``Power decoding'' approach by Schmidt et al. \cite{schmidt06rs}.
Though this did not need the generalisation of the $\nu$-weight, that was included with the outlook of decoding Algebraic Geometric codes, inspired by the approach of Brander \cite{brander10}.

\section*{Acknowledgments}
The author would like to thank Peter Beelen for valuable advice on this work and its exposition.
The author gratefully acknowledges the support from the Otto M\o nsted's Fund and the Idella Fund, as well as the support from the Danish National Research Foundation and the National Science Foundation of China (Grant No.11061130539) for the Danish-Chinese Center for Applications of Algebraic Geometry in Coding Theory and Cryptography.

\bibliographystyle{IEEEtran}
\bibliography{../../tex/bibtex_old}

\end{document}